\documentclass[11.5pt]{scrartcl}

\usepackage{xy-format}
\usepackage{xy-theorem}
\usepackage{fancyhdr}
\usepackage{amsmath}

\newcommand{\R}{\mathbb{R}}
\newcommand{\C}{\mathbb{C}}
\newcommand{\what}[1]{\widehat{#1}}
\newcommand{\dd}{\partial}
\newcommand{\sd}{\not\!\partial}
\newcommand{\tf}{\mathrm{tf}}
\newcommand{\gtf}{\gamma\mathrm{tf}}
\newcommand{\Pp}{P_+^{\perp}}
\newcommand{\Pm}{P_-^{\perp}}
\newcommand{\Stf}{S^2_0(T^*\dd X)}
\DeclareMathOperator{\img}{im}
\DeclareMathOperator{\rank}{rank}
\DeclareMathOperator{\Reop}{Re}
\DeclareMathOperator{\diag}{diag}

\title{A Linearized Obstruction to the Supersymmetric Extension of Conformal Boundary Conditions in Euclidean Gravity}

\author[1]{Xingyang Yu}
\affil[1]{Department of Physics, Virginia Tech, Blacksburg, VA 24061, USA}

\date{}

\begin{document}
\maketitle

\begin{abstract}
The conformal boundary condition for gravity fixes the boundary conformal class and the mean
curvature, leaving the trace-free extrinsic curvature free as the conjugate response. It is the
boundary-value form of York's conformal decomposition of gravitational data
\cite{York:1972sj,York:1973ia}, shown to be well posed for Einstein metrics by Anderson
\cite{Anderson:2006lqb}. Witten identified it as the elliptic replacement for the ill-posed
Dirichlet condition in the finite-boundary perturbative Euclidean gravitational path integral
\cite{Witten:2018lgb}. We show that this perturbative construction admits no
half-supersymmetric extension in linearized minimal supergravity. For
fixed conformal bosonic data, no half-dimensional gravitino boundary condition (local or
pseudodifferential, APS-type included, with any compatible ghost condition at highest-derivative
order) closes the full preserved chiral supersymmetry. Supersymmetry first selects the natural
local chiral gravitino datum. Acting back on this datum then produces the trace-free extrinsic
curvature, precisely the response that the conformal prescription leaves unfixed. The
obstruction is therefore not the failure of a particular elliptic ansatz: even the chiral/Robin
completion that is LS-elliptic and BRST-compatible at highest-derivative order would impose
Dirichlet control on a Neumann response. The obstruction is pointwise in tangential momentum and
survives compensating gauge transformations. It is a linearized, highest-derivative obstruction,
not a global or nonlinear no-go; nonlinear supercovariant boundary terms may evade it by tying
the trace-free extrinsic curvature to gravitino bilinears.
\end{abstract}

\paragraph{AI Use Statement.}
The numerical verification scripts listed in the Reproducibility note, together with part of
the manuscript preparation, were carried out with substantial assistance from AI tools,
including Anthropic's Claude Code and OpenAI's Codex. The author directed the work, reviewed
and edited the manuscript, checked the analytic arguments, scripts, citations, and conclusions,
and takes full responsibility for the final content.

\clearpage
\tableofcontents

\section{Introduction}

In Euclidean quantum gravity with a finite boundary, the choice of boundary data is part of
the definition of the perturbative problem. Around a classical saddle, the gauge-fixed
graviton and ghost fluctuations should form a Fredholm elliptic boundary-value problem; this
is what underlies the heat-kernel expansion, the one-loop determinant, and a controlled mode
expansion \cite{Vassilevich:2003xt}. The standard Dirichlet choice, fixing the induced metric $h_{ab}$ as for a scalar
field, does not have this property. Failure of the Lopatinski--Shapiro (LS) condition
\cite{Lopatinskii:1953,Shapiro:1953,BaerBallmann:2012} means that the principal boundary symbol has nonzero modes localized near the boundary at
arbitrarily large tangential momentum. Witten showed that the gauge-fixed linearized Einstein
operator with the induced metric fixed has precisely this failure \cite{Witten:2018lgb}, so
the perturbative expansion around a saddle is not well posed with Dirichlet metric data.

The elliptic replacement is to fix less data. One fixes only the conformal class $[h_{ab}]$
of the boundary metric and the mean curvature $K$, leaving the traceless extrinsic curvature
$K_{ab}^{\tf}$ free as the conjugate response. This is the boundary-value version of York's
conformal decomposition of gravitational data
\cite{York:1972sj,York:1973ia}. Following Witten, we call these \emph{conformal} boundary
data. Apart from standard names such as Brown--York \cite{Brown:1992br} and
Gibbons--Hawking--York \cite{Gibbons:1976ue,York:1972sj}, ``York''
below refers to the underlying conformal/mean-curvature polarization: $[h_{ab}]$ and $K$
are fixed as sources, while
$K_{ab}^{\tf}$ is left as the conjugate momentum. This conformal boundary condition is
elliptic (its well-posedness for Einstein metrics was established by Anderson
\cite{Anderson:2006lqb}) and it is the natural boundary data for the gravitational path
integral; recent work has developed these conformal boundary conditions considerably,
including a one-parameter family of well-posed variants \cite{Liu:2024ymn} and their
finite-boundary thermodynamics \cite{Anninos:2023epi,Banihashemi:2025qqi}. Throughout,
Greek indices $\mu,\nu$ are bulk indices, Latin indices $a,b$ are tangent to $\dd X$, and
$\perp$ denotes the normal direction. On flat half-space $\R^4_+=\{x^\perp\ge0\}$, in
de~Donder gauge, the principal linearized conformal boundary rows are
\begin{equation}
\label{eq:york}
h_{ab}^{\tf}\big|=0,\qquad
\mathcal K(h)\big|:=\dd_\perp\tau-2\dd^a h_{\perp a}=0,\qquad
T_\mu(h)\big|=0,
\end{equation}
with $h_{ab}=h_{ab}^{\tf}+\tfrac13\delta_{ab}\tau$, $\tau=h^a{}_a$, $h=h^\rho{}_\rho$, and
$T_\mu(h)=\dd^\nu h_{\mu\nu}-\tfrac12\dd_\mu h$ the de~Donder gauge function. The feature
that matters below is the last part of the prescription: $K_{ab}^{\tf}$, the momentum
conjugate to $[h_{ab}]$, is left \emph{free}, as a Neumann response, the quasilocal
boundary stress-tensor response to varying the metric, in the Brown--York sense
\cite{Brown:1992br}.

We now ask whether this perturbative boundary problem can be made compatible with supersymmetry.
For BPS observables and localization \cite{Pestun:2007rz}, or supersymmetric holographic renormalization
\cite{BenettiGenolini:2016tsn}, it is natural to require the path integral to preserve half
of the supersymmetry. Throughout, ``half'' means the \emph{full} preserved chiral space
$\img\Pp$; the weaker question of preserving a single fixed complex supercharge is not
addressed here. The boundary problem then has to satisfy three requirements at once. It must
keep the elliptic conformal bosonic data \eqref{eq:york}; it must be compatible with the
BRST gauge fixing \cite{Becchi:1975nq,Tyutin:1975qk} of diffeomorphisms and local supersymmetry; and it must impose a gravitino
boundary condition closed under the preserved supersymmetry. We call such a completion a
\emph{super-York} boundary condition: the name marks a supersymmetric version of the York
polarization underlying the conformal condition. Although the bosonic data are conformal,
we do not call the completion ``superconformal'': in boundary conformal field theory,
\emph{superconformal boundary conditions} already denote conditions preserving part of the
superconformal algebra \cite{Gaiotto:2008sa}. In this paper, ``conformal'' refers instead to
fixing the conformal \emph{class} of the boundary metric.

This question is part of a long Euclidean boundary-condition literature. Gauge-invariant
boundary operators for Euclidean gravity, including the de~Donder boundary operator with
projection and nilpotent pieces, were studied by Avramidi--Esposito--Kamenshchik
\cite{Avramidi:1996ae}. Local gravitino boundary conditions and their one-loop consequences
were analyzed by Esposito \cite{Esposito:1996kv,Esposito:1995xy,Esposito:1997wt}, in parallel with the spectral conditions of
D'Eath--Esposito \cite{DEath:1991opx}. Moss and collaborators emphasized the tension between BRST-invariant
boundary conditions \cite{Moss:1996ip} and strong ellipticity in quantum gravity, including
extrinsic-curvature boundary data \cite{Moss:2013vh}. For supergravity itself,
van~Nieuwenhuizen--Vassilevich derived orbits of boundary conditions closed under local
symmetries and BRST transformations, and showed that purely local Gibbons--Hawking-type
\cite{Gibbons:1976ue} conditions break local supersymmetry \cite{vanNieuwenhuizen:2005kg}. Our statement is
narrower than these global analyses: it is a local high-momentum obstruction for fixed
conformal bosonic data. The novelty is not the broad fact that familiar local
supergravity boundary conditions can fail to preserve local supersymmetry; it is the
identification, for the conformal boundary problem, of the gauge-invariant obstruction with
$K_{ab}^{\tf}$, the Neumann response conjugate to the boundary conformal class.

We show that Witten's elliptic conformal boundary problem for perturbative Euclidean gravity has
no separated linear half-supersymmetric completion in minimal supergravity. The mechanism is
simple. Supersymmetry ties the boundary metric to the gravitino,
$\delta_\varepsilon h_{ab}\sim\bar\varepsilon\,\gamma_{(a}\psi_{b)}$; acting once more, it
ties the gravitino back to a derivative of the metric, namely to the extrinsic curvature.
Concretely, the supersymmetry
variation of the natural chiral gravitino datum is the boundary normal--tangential spin
connection $\omega_a{}^{b\perp}$, whose gauge-invariant content is exactly the traceless
extrinsic curvature $K_{ab}^{\tf}$. But $K_{ab}^{\tf}$ is precisely what \eqref{eq:york}
leaves free. Supersymmetry wants to fix it (to treat $K_{ab}^{\tf}$ as Dirichlet data tied
to the gravitino) while the conformal condition insists on keeping it as a Neumann
response. A coordinate $[h_{ab}]$ and its conjugate momentum $K_{ab}^{\tf}$ cannot both be
held fixed; this is the Dirichlet--Neumann mismatch behind the obstruction.

Technically, the proof uses the standard local high-momentum test for elliptic boundary
conditions. We replace the boundary by its tangent flat half-space, evaluate all coefficients
at the boundary point, discard lower-derivative terms, Fourier transform along the boundary,
and check the resulting finite-dimensional linear algebra at each nonzero tangential momentum.
In the PDE literature this finite-dimensional object is the principal boundary symbol. We
then ask, for \emph{any} choice $V(k)$ of the half of the gravitino boundary data left
free, local, or of Atiyah--Patodi--Singer (APS)/Calderon type, meaning nonlocal spectral
projectors \cite{Atiyah:1975jf,Calderon:1963,Seeley:1966,BaerBallmann:2012}, whether it can close the algebraic half-supersymmetry
variation while keeping \eqref{eq:york} fixed. Theorem~\ref{thm:main} gives a negative answer.
The two natural possibilities fail for different reasons. The
spectral (APS) condition, which is natural and self-adjoint for Dirac-type problems and was
advocated by D'Eath and Esposito for the one-loop gravitino \cite{DEath:1991opx}, breaks the
boundary chirality required by supersymmetry. The local chiral condition passes this necessary
conformal-class test and satisfies the ellipticity and BRST checks described below, but its
reverse supersymmetry variation is precisely the unfixed $K_{ab}^{\tf}$ response.

This does not contradict the existence of supersymmetric supergravity boundaries. The pair
$(K_{ab}^{\tf},\Pm\psi_a)$ is the bottom component of a boundary extrinsic-curvature
multiplet \cite{Belyaev:2010as}, the boundary analogue of the bulk supercurrent multiplet,
in which the stress tensor and the supercurrent sit together \cite{Komargodski:2010rb}.
In the Lorentzian, fully nonlinear theory
Belyaev--van~Nieuwenhuizen and Belyaev--Pugh
\cite{Belyaev:2005rt,Belyaev:2007bg,Belyaev:2008ex,Belyaev:2010as} build
half-supersymmetric boundaries precisely by supercovariantizing the spin connection,
$\what{\omega}_a{}^{b\perp}=\omega_a{}^{b\perp}+\psi\psi$, and adding a boundary action
($\oint K$ for gravity, $\oint\bar\psi\gamma\psi$ for the gravitino) that ties
$K_{ab}^{\tf}$ to a gravitino bilinear. Our result is the linearized shadow that explains
why such nonlinear terms are the natural way to evade the present obstruction: in the local
high-momentum boundary test, $K_{ab}^{\tf}$ is free
and the Dirichlet--Neumann mismatch obstructs closure. The statement is a linearized, flat-half-space,
highest-derivative statement for the gauge-fixed theory; being pointwise
in the tangential momentum $k\neq0$ it covers pseudodifferential (APS-type) conditions, but
it is not a global or nonlinear claim. The ancillary scripts reproduce the finite-dimensional
linear algebra in several gamma representations and momenta; the result-to-script map and the
reproduction command are collected in the Reproducibility note at the end of
Section~\ref{sec:disc}.

\paragraph{Organization of the paper.}
Section~\ref{sec:bosonic-york} rederives Witten's local half-space test for Dirichlet and
conformal boundary data, emphasizing the polarization in which $K_{ab}^{\tf}$ is left
free. Section~\ref{sec:sugra-dictionary}
collects the boundary-supergravity dictionary needed for the argument. Section~\ref{sec:nogo}
contains the local no-go proof: it sets up the linearized flat-half-space model, proves that
the forward variation selects the chiral gravitino datum, records the chiral/Robin completion
showing that the obstruction is not caused by a simple ellipticity or BRST-counting failure,
identifies the reverse variation with $K_{ab}^{\tf}$, and states the main theorem.
Section~\ref{sec:otherd} comments on other dimensions.
Section~\ref{sec:disc} then discusses the physical meaning and outlook.

\section{Bosonic Conformal Boundary Data}
\label{sec:bosonic-york}

We begin by redoing the local calculation behind Witten's conformal boundary condition
\cite{Witten:2018lgb}, in the notation used later. This is the one bosonic fact on which the
supergravity argument rests. The conformal boundary condition is the alternative to Dirichlet
metric data that passes the local high-momentum ellipticity test, and it does so precisely by
leaving $K_{ab}^{\tf}$ unfixed. In the perturbative path-integral motivation, this test is the
local criterion for the Gaussian fluctuation integral around a saddle to be well defined.

\paragraph{The local LS test.}
For a second-order Laplace-type operator on a half-space, the relevant LS modes are plane
waves that oscillate along the boundary and decay into the bulk. Thus take
\[
X=\R^4_+=\{x^\perp\ge0\},\qquad \dd X=\R^3,
\]
Fourier transform along the boundary, and write a decaying graviton mode as
\begin{equation}
\label{eq:witten-mode}
h_{\mu\nu}=H_{\mu\nu}\,
\exp\!\left(i k_a y^a-|k|x^\perp\right),\qquad k\neq0 .
\end{equation}
At this principal-symbol level
\[
\dd_a\mapsto ik_a,\qquad \dd_\perp\mapsto-|k|.
\]
The gauge-fixed linearized Einstein operator is Laplace-type in de~Donder gauge, so every
constant symmetric tensor $H_{\mu\nu}$ gives such a decaying solution of the principal
interior equation. Ellipticity of a boundary condition says that no nonzero decaying solution
of this form should obey all homogeneous boundary rows.

\paragraph{Dirichlet metric data fail.}
The gravitational analogue of scalar Dirichlet data would be
\[
h_{ab}\big|=0,\qquad T_\mu(h)\big|=0,
\]
where the de~Donder rows are imposed because the gauge-fixed problem must also be compatible
with the diffeomorphism ghosts. Put
\[
H_{ab}=0,\qquad q_a:=H_{\perp a},\qquad \beta:=H_{\perp\perp}.
\]
Since the trace is then $H=\beta$, the boundary-symbol equations $T_\perp=0$ and $T_a=0$
become
\begin{align}
\label{eq:dirichlet-symbol}
ik^a q_a-\frac12 |k|\beta&=0,\\
-|k|q_a-\frac{i}{2}k_a\beta&=0 .
\end{align}
They have a nonzero solution for every $\beta$:
\begin{equation}
\label{eq:dirichlet-badmode}
q_a=-\frac{i}{2|k|}k_a\beta .
\end{equation}
Thus the full Dirichlet metric condition admits a boundary-localized decaying graviton mode
for every nonzero tangential momentum. Since $|k|$ can be made arbitrarily large, this is the
LS failure: the boundary-value problem has too many short-distance modes concentrated near
the boundary. This is Witten's local reason why the Dirichlet metric condition is not
elliptic. The same phenomenon can be interpreted geometrically as
``boundary-moving'' diffeomorphisms: a normal displacement of the boundary can preserve the
induced metric to first order while still producing nontrivial bulk metric fluctuations.

\paragraph{Conformal data pass.}
The conformal boundary condition weakens the metric row and replaces the missing scalar
condition by fixing the trace of the extrinsic curvature. In York variables,
\[
h_{ab}=e^{2\varphi}\hat h_{ab},\qquad \det\hat h=1,\qquad
K_{ab}=K_{ab}^{\tf}+\frac13h_{ab}K .
\]
Infinitesimally, fixing the boundary conformal class means fixing only
$h_{ab}^{\tf}$; fixing $K$ gives, at highest-derivative order,
\[
2\,\delta K=\dd_\perp\tau-2\dd^a h_{\perp a},\qquad
\tau=h^a{}_a .
\]
Thus the local boundary rows are precisely those displayed in \eqref{eq:york}:
\[
h_{ab}^{\tf}\big|=0,\qquad
\mathcal K(h)\big|=0,\qquad
T_\mu(h)\big|=0 .
\]
To check ellipticity, write the tangential part of the amplitude as a pure trace,
\[
H_{ab}=\delta_{ab}r,\qquad \tau=3r,
\]
and keep the same notation $q_a=H_{\perp a}$ and $\beta=H_{\perp\perp}$. The mean-curvature
row gives
\begin{equation}
\label{eq:york-krow}
2ik^a q_a+3|k|r=0 .
\end{equation}
The de~Donder rows are
\begin{align}
\label{eq:york-tperp}
ik^a q_a-\frac12 |k|\beta+\frac32 |k|r&=0,\\
\label{eq:york-ta}
-|k|q_a-\frac{i}{2}k_a\beta-\frac{i}{2}k_a r&=0 .
\end{align}
Comparing \eqref{eq:york-krow} with \eqref{eq:york-tperp} first gives $\beta=0$. Then
\eqref{eq:york-ta} gives
\[
q_a=-\frac{i}{2|k|}k_a r .
\]
Substituting this back into \eqref{eq:york-krow} gives
\[
4|k|r=0 .
\]
Since $k\neq0$, we get $r=0$, hence $q_a=0$ and $\beta=0$. Therefore the only decaying
solution satisfying the conformal boundary rows is the zero solution. This is the local
symbol calculation behind the statement that the conformal boundary condition is elliptic.

\paragraph{The polarization.}
The same calculation identifies which gravitational momentum has been left free. At linear
order, with the same choice of normal coordinate,
\begin{equation}
\label{eq:lin-extrinsic}
2\,\delta K_{ab}
=\dd_\perp h_{ab}-\dd_a h_{\perp b}-\dd_b h_{\perp a},\qquad
2\,\delta K=\dd_\perp\tau-2\dd^a h_{\perp a}.
\end{equation}
The mean-curvature row fixes the trace $\delta K$. The rows $h_{ab}^{\tf}=0$ fix the coordinate
part of the boundary conformal metric. But no boundary row fixes the normal derivative of
$h_{ab}^{\tf}$, or equivalently the trace-free part of $\delta K_{ab}$ modulo the tangential
shift terms in \eqref{eq:lin-extrinsic}. This is not a missing condition; it is the intended
choice of boundary polarization:
\[
\text{conformal data fix }[h_{ab}]\text{ and }K,\qquad
\text{but leave }K_{ab}^{\tf}\text{ free.}
\]
In the Brown--York variation of the gravitational action, the response conjugate to a
traceless variation of the boundary metric is built from precisely this
$K_{ab}^{\tf}$ component \cite{Brown:1992br}. A Dirichlet boundary, a brane boundary, or an
Israel junction condition \cite{Israel:1966rt} may impose a relation on $K_{ab}^{\tf}$; that is a different
ensemble. The supergravity question in this paper is sharper: can one preserve boundary
supersymmetry while keeping the York polarization underlying Witten's conformal condition, in which
$K_{ab}^{\tf}$ remains a Neumann response?

\section{A Minimal Boundary-Supergravity Dictionary}
\label{sec:sugra-dictionary}

This section collects the minimal supergravity language needed for the no-go. We do not
review all formulations of Euclidean supergravity with boundaries; we only record the
boundary facts used in the symbol calculation. The forward and reverse supersymmetry
variations are previewed here heuristically; both are established rigorously, as part of the
proof, in Section~\ref{sec:nogo}.

\paragraph{Bulk fields and principal supersymmetry.}
We expand around flat space with vanishing background gravitino. At the linearized,
highest-derivative level the fields are the metric perturbation and the Rarita--Schwinger
field \cite{Rarita:1941mf},
\[
(h_{\mu\nu},\psi_\mu).
\]
Auxiliary fields of old- or new-minimal supergravity, and lower-order curvature or
cosmological-constant terms, do not enter the principal boundary symbol considered here. The
rigid supersymmetry transformations, in the conventions of \cite{VanNieuwenhuizen:1981ae}, are
\begin{equation}
\label{eq:sugra-dictionary-susy}
\delta_\varepsilon h_{\mu\nu}
=\bar\varepsilon\,\gamma_{(\mu}\psi_{\nu)},\qquad
\delta_\varepsilon\psi_\mu
=\frac14\,\omega_\mu{}^{mn}(h)\gamma_{mn}\varepsilon ,
\end{equation}
where
\[
\omega_\mu{}^{mn}(h)=\frac12(\dd^n h_\mu{}^m-\dd^m h_\mu{}^n),\qquad
\gamma_{mn}:=\frac12[\gamma_m,\gamma_n].
\]
These two transformation laws already contain the obstruction. The first maps a boundary
metric condition to a condition on the tangential gravitino. The second maps a gravitino
condition back to a normal derivative of the metric, hence to extrinsic curvature.

\paragraph{Normal chirality.}
A boundary gives a preferred Clifford element, the normal gamma matrix. We write
\[
P_\pm^\perp=\frac12(1\pm\gamma^\perp).
\]
This is not four-dimensional chirality; it is chirality with respect to the boundary normal.
Tangential gamma matrices flip it:
\begin{equation}
\label{eq:dictionary-flip}
\gamma_aP_+^\perp=P_-^\perp\gamma_a,\qquad
\gamma_aP_-^\perp=P_+^\perp\gamma_a .
\end{equation}
We choose the preserved half-supersymmetry to obey
\[
P_+^\perp\varepsilon=\varepsilon .
\]
Changing this choice only exchanges the two signs everywhere.

\paragraph{A spin-1/2 warm-up: rotation invariance fixes the projector.}
Before the gravitino, consider what already fixes the projector for a spin-1/2 field. A local
half-rank condition $P\psi|=0$ selects a rank-two projector $P$ on the spinor index, and there
are many such $P$. The physical constraint is boundary rotation invariance,\footnote{The author
is grateful to Edward Witten for raising the spin-1/2 question that motivates this warm-up and
for emphasizing that a boundary projector must be rotation invariant.} that $P$ commute with the
tangential rotation generators,
\begin{equation}
\label{eq:rot-inv}
[P,\gamma_{ab}]=0 .
\end{equation}
In four dimensions the elements of the Clifford algebra commuting with all tangential
$\gamma_{ab}$ are spanned by
\begin{equation}
\label{eq:commutant}
\{\,1,\ \gamma^\perp,\ \gamma_5,\ \gamma_5\gamma^\perp\,\},
\end{equation}
so the rotation-invariant local projectors are built from the normal element $\gamma^\perp$, up
to a chiral-angle rotation by $\gamma_5$. Three natural candidates result,
\begin{equation}
\label{eq:three-projectors}
\underbrace{\tfrac12(1\pm\gamma^\perp)}_{\text{bag/MIT}},\qquad
\underbrace{\tfrac12(1\pm\gamma_5)}_{\text{4d chirality}},\qquad
\underbrace{\tfrac12(1\pm A)}_{\text{APS}},\qquad
A=i\gamma^\perp\gamma^a\hat k_a ,
\end{equation}
classified by rotation invariance, LS-ellipticity, and locality:
\begin{center}
\begin{tabular}{lccc}
\toprule
local half-rank projector & rotation inv. & LS-elliptic & local \\
\midrule
$\tfrac12(1\pm\gamma^\perp)$ \quad(bag/MIT) & \checkmark & \checkmark & \checkmark \\
$\tfrac12(1\pm\gamma_5)$ \quad(4d chirality) & \checkmark & $\times$ & \checkmark \\
$\tfrac12(1\pm A)$ \quad(APS, spectral) & \checkmark & \checkmark & $\times$ \\
\bottomrule
\end{tabular}
\end{center}
Only the bag projector passes all three. The bag/MIT projector
$P_\pm^\perp=\tfrac12(1\pm\gamma^\perp)$, the normal-chirality condition $\gamma^\perp\psi|=\pm\psi$
of the preceding paragraph (the boundary form of the MIT bag condition), is local and LS-elliptic. The four-dimensional
chirality projector $\tfrac12(1\pm\gamma_5)$ is local but not elliptic: it commutes with $A$, so
the space of decaying boundary modes splits by chirality and one component is left
unconstrained, a zero mode at every tangential momentum. The APS spectral projector
$\tfrac12(1\pm A)$ is elliptic but nonlocal \cite{Atiyah:1975jf}. This is the spinor form of the
local-versus-spectral dichotomy for Euclidean fermion boundary conditions
\cite{DEath:1991opx,Luckock:1990xr}. What the rest of this section shows is that the gravitino
inherits the same normal projector $P_-^\perp$, and that supersymmetry leaves no freedom in the
choice: it is not one rotation-invariant option among several but is forced by closure of the
conformal class.

\paragraph{Why the tangential gravitino must be chiral.}
The conformal metric row is $h_{ab}^{\tf}|=0$. Its supersymmetry variation is
\[
\delta_\varepsilon h_{ab}^{\tf}
=\left(\bar\varepsilon\,\gamma_{(a}\psi_{b)}\right)^{\tf}.
\]
Because $\bar\varepsilon=\bar\varepsilon P_+^\perp$ in the complexified convention used
below, and because of \eqref{eq:dictionary-flip}, this expression sees the
$P_-^\perp$ part of the tangential gravitino. The gamma-trace contribution to $\psi_a$ gives
only a trace in $ab$, so it drops out of the trace-free conformal row. Thus the local way to
keep the conformal class fixed under all preserved supersymmetries is
\begin{equation}
\label{eq:dictionary-chiral-bc}
P_-^\perp\psi_a\big|=0 ,
\end{equation}
leaving the $P_+^\perp$ part of $\psi_a$ as the unconstrained tangential gravitino datum. This
is the forward direction of the argument. The APS/spectral projector is natural for a
Dirac-type operator, but it projects using the tangential Dirac symbol
$i\gamma^\perp\gamma^a\hat k_a$; that operator mixes the two normal chiralities and therefore
does not impose \eqref{eq:dictionary-chiral-bc}.

\paragraph{What the reverse variation measures.}
Now vary the chiral row \eqref{eq:dictionary-chiral-bc}. Splitting the spin connection in
\eqref{eq:sugra-dictionary-susy} into tangential--tangential and normal--tangential pieces,
\[
\delta_\varepsilon\psi_a
=\frac14\,\omega_a{}^{bc}\gamma_{bc}\varepsilon
+\frac12\,\omega_a{}^{b\perp}\gamma_{b\perp}\varepsilon .
\]
The first term contains two tangential gamma matrices and preserves normal chirality, so it
is killed by $P_-^\perp$ when $\varepsilon\in\img P_+^\perp$. The second term contains one
tangential gamma matrix and reaches the opposite normal chirality. Therefore
\begin{equation}
\label{eq:dictionary-reverse}
\delta_\varepsilon\!\left(P_-^\perp\psi_a\big|\right)
=\frac12\,\omega_a{}^{b\perp}\big|\,\gamma_b\varepsilon .
\end{equation}
At the same linearized highest-derivative level,
\begin{equation}
\label{eq:dictionary-omega-k}
\omega_a{}^{b\perp}
=\frac12(\dd_\perp h_a{}^b-\dd^b h_{a\perp})
=K_a{}^b+\frac12\dd_a h^b{}_\perp ,
\end{equation}
with the convention for $K_{ab}$ in \eqref{eq:lin-extrinsic}. The second term in
\eqref{eq:dictionary-omega-k} is the tangential derivative of the shift component
$h_{a\perp}$; in the symbol calculation it is a compensating local-supersymmetry gauge image.
After quotienting this gauge image and after the mean-curvature row has fixed the trace $K$, the
gauge-invariant content of \eqref{eq:dictionary-reverse} is exactly
$K_{ab}^{\tf}$.

\paragraph{Boundary multiplet intuition.}
Thus, modulo gauge and trace pieces, the preserved boundary supersymmetry pairs
\[
P_-^\perp\psi_a
\qquad\longleftrightarrow\qquad
K_{ab}^{\tf}.
\]
This gives the geometric content of the theorem. The bosonic conformal ensemble deliberately leaves
$K_{ab}^{\tf}$ free as the Neumann response conjugate to the boundary conformal class. But the
fermionic chiral row required by the forward variation has a reverse variation proportional
to that same $K_{ab}^{\tf}$. Preserving supersymmetry would therefore require the boundary
condition to control a variable that the conformal polarization is designed not to fix.

\paragraph{Relation to familiar Dirichlet supergravity boundaries.}
A standard half-BPS supergravity boundary usually fixes the induced boundary
frame, imposes one normal-chirality condition on the tangential gravitino, and adds the
Gibbons--Hawking--York term together with its fermionic partner
\cite{Belyaev:2005rt,Belyaev:2007bg}. In that setting the
extrinsic-curvature response is part of the boundary stress-tensor/supercurrent multiplet
controlled by the boundary action. That is a natural supersymmetric Dirichlet-type
ensemble, but it is not Witten's conformal ensemble: here the full boundary metric is not fixed,
and $K_{ab}^{\tf}$ is meant to remain the response conjugate to the conformal class. The
calculation below asks whether one can keep that bosonic polarization and still impose a
linear half-supersymmetric gravitino boundary condition.

For this reason, the theorem should not be read as a no-go for all half-BPS supergravity
boundaries: the nonlinear, supercovariant boundary constructions of
Belyaev--van~Nieuwenhuizen and Belyaev--Pugh
\cite{Belyaev:2005rt,Belyaev:2007bg,Belyaev:2010as} (which supercovariantize the
normal--tangential connection and add boundary action terms) may evade the linear
obstruction; we return to this in Section~\ref{sec:disc}. The result below concerns the purely
linear attempt to keep Witten's bosonic conformal polarization and add an independent local
half-supersymmetric gravitino boundary condition.

\section{The Local No-Go Argument}
\label{sec:nogo}

We now prove the local obstruction advertised above. The point of the argument is not to
classify supersymmetric boundary conditions for Euclidean supergravity. It is to test the
particular perturbative ensemble singled out by Witten's conformal bosonic data: the boundary
conformal class is fixed, the trace of the extrinsic curvature supplies the missing scalar
condition, and the traceless extrinsic curvature $K_{ab}^{\tf}$ remains the response variable.
We therefore work at the level of the boundary principal symbol, where ellipticity and
high-momentum zero modes are decided.

The proof has two directions. First, preservation of the bosonic conformal rows by a boundary
supersymmetry forces the tangential gravitino datum to be chiral. Second, reversing the
supersymmetry variation shows that no half-dimensional local fermionic datum can avoid the
normal--tangential spin connection, whose gauge-invariant core is precisely
$K_{ab}^{\tf}$. Thus the same datum which must remain free in Witten's bosonic polarization is
forced back into the boundary condition by linear supersymmetry.

\subsection{Setup}
\label{sec:setup}

\paragraph{Model.} We work on the flat Euclidean half-space
\begin{equation}
X=\R^4_+=\{x^\perp\ge0\},\qquad
\dd X=\R^3,\qquad
g_{\mu\nu}^{(0)}=\delta_{\mu\nu},\qquad
\psi^{(0)}_\mu=0 .
\end{equation}
The complexified linearized fields are
\begin{equation}
(h_{\mu\nu},\psi_\mu),\qquad
h_{\mu\nu}=h_{\nu\mu},
\end{equation}
with $\psi_\mu$ a vector--spinor. Thus the theorem is a statement about the complexified
boundary symbol; a choice of Euclidean reality condition or integration contour is an
additional structure. No off-shell auxiliary fields (old- or new-minimal) are included; they
would add only lower-order, non-principal terms to $\delta_\varepsilon\psi_\mu$, and any
boundary data they carry lies outside the theorem's highest-derivative scope.

Euclidean gamma matrices are Hermitian and obey
\begin{equation}
\{\gamma^\mu,\gamma^\nu\}=2\delta^{\mu\nu},\qquad
\gamma^\perp=\gamma^4,\qquad
\gamma^a\quad(a=1,2,3),\qquad
P_\pm^{\perp}=\frac12(1\pm\gamma^\perp).
\end{equation}
Here $\gamma^\perp$ is the boundary-normal gamma matrix and the $\gamma^a$ are tangential.
Multi-index gamma matrices are antisymmetrized Clifford products; in particular
\begin{equation}
\gamma_{mn}:=\frac12[\gamma_m,\gamma_n].
\end{equation}

\paragraph{Symmetries.} The linearized gauge symmetries are diffeomorphisms and local
supersymmetry:
\begin{equation}
\delta_\xi h_{\mu\nu}=\dd_\mu\xi_\nu+\dd_\nu\xi_\mu,\qquad
\delta_\zeta\psi_\mu=\dd_\mu\zeta .
\end{equation}
The rigid supersymmetry we attempt to preserve acts as
\begin{equation}
\label{eq:susy}
\delta_\varepsilon h_{\mu\nu}=\bar\varepsilon\,\gamma_{(\mu}\psi_{\nu)},
\qquad
\delta_\varepsilon\psi_\mu=\tfrac14\,\omega_\mu{}^{mn}(h)\,\gamma_{mn}\,\varepsilon,
\end{equation}
with constant $\varepsilon$ obeying the half-projection
\begin{equation}
\Pp\varepsilon=\varepsilon .
\end{equation}
The linearized spin connection is
\begin{equation}
\omega_\mu{}^{mn}(h)=\frac12(\dd^n h_\mu{}^m-\dd^m h_\mu{}^n).
\end{equation}
The frame indices $m,n$ run over all four directions, tangential and normal; the full
connection, not a tangential truncation, is used throughout. For the symbol calculation we use
the Hermitian Dirac pairing
\begin{equation}
\bar\varepsilon=\varepsilon^\dagger,\qquad
\Pp\varepsilon=\varepsilon\ \Longrightarrow\ \bar\varepsilon=\bar\varepsilon\Pp ,
\end{equation}
where the implication uses the Hermiticity of $\Pp$. If one instead
packages Euclidean supergravity spinors by a charge-conjugation or symplectic-Majorana
doubling, the induced left projector is exchanged with $\Pm$; this conjugates the notation
$\Pp\leftrightarrow\Pm$ but leaves the rank and obstruction statements unchanged. No
positivity or real-structure property is used below. Thus ``preserving half-supersymmetry''
means preserving the two-complex-dimensional algebraic subspace
\begin{equation}
\img\Pp,\qquad \dim_\C\img\Pp=2,
\end{equation}
of the complexified symbol problem. Selecting a real Euclidean contour, or imposing a doubled
symplectic-Majorana reality condition before continuation, may change which representative
of the pair $(\Pp,\Pm)$ is called preserved, but it does not change the finite-dimensional
rank obstruction.

\paragraph{Gauge-fixed operators and ellipticity.} In de~Donder gauge the linearized
Einstein operator is Laplace-type. The gauge-fixed Rarita--Schwinger operator is denoted
$F$; at principal-symbol level it is Dirac-type \cite{VanNieuwenhuizen:1981ae},
\begin{equation}
F_\mu{}^\nu\sim i\sd\,\delta_\mu{}^\nu,\qquad
\sd:=\gamma^\rho\dd_\rho .
\end{equation}
The minimal ghost operators have principal parts
\begin{equation}
c_\mu:\ -\dd^2,\qquad
\eta:\ L\sim\sd ,
\end{equation}
where $c_\mu$ is the anticommuting diffeomorphism ghost and $\eta$ is the commuting-spinor
local-supersymmetry ghost. Since the metric equation is second order and the gravitino
equation is first order, the boundary problem has mixed differential order (a
Douglis--Nirenberg system in PDE terminology \cite{DouglisNirenberg:1955,AgmonDouglisNirenberg:1959}). The bosonic and fermionic orders are counted
with weights
\begin{equation}
w_{\mathrm{bos}}=2,\qquad
w_{\mathrm{ferm}}=1,\qquad
\sd^2=-\dd^2 .
\end{equation}

The Lopatinski--Shapiro (LS) test is the local ultraviolet test for such a boundary-value
problem. One first works at a boundary point in coordinates where the boundary looks like the
flat half-space, then replaces all slowly varying coefficients by their values at that point,
and finally drops lower-derivative terms. Thus, for LS purposes, the question becomes a
constant-coefficient problem on $\R^4_+$. One Fourier transforms along the boundary while the
normal direction remains a half-line problem:
\begin{equation}
\dd_a\mapsto ik_a,\qquad
u(x^\perp,y)=u_0\,e^{ik\cdot y-\lambda x^\perp},\qquad
k\neq0,\quad \Reop\lambda>0 .
\end{equation}
The LS question is whether the highest-derivative part of the boundary operator is an
isomorphism on this decaying-solution space \cite{BaerBallmann:2012}. This finite-dimensional
map is what we call the principal boundary symbol.

For the Dirac-type $F$, the decaying space is the $+1$ eigenspace of
\begin{equation}
A(\hat k)=i\gamma^\perp\gamma^a\hat k_a,\qquad
\hat k=\frac{k}{|k|},\qquad
A^2=\mathbf 1 .
\end{equation}
The corresponding APS spectral projectors are
\begin{equation}
\Pi_\pm=\frac12(\mathbf 1\pm A).
\end{equation}

The BRST statements below are likewise statements about this highest-derivative part of the
gauge fields and minimal ghosts $(c_\mu,\eta)$ only. We do not construct a full quantum BRST boundary complex:
antighosts, Nakanishi--Lautrup or other auxiliary fields, lower-order gauge-fixing terms,
and global restrictions on boundary-compatible gauge parameters are outside the theorem.
When we quotient by a compensating local supersymmetry in the no-go proof, we use the full
highest-derivative gauge variation, explicitly an enlarged quotient.

\paragraph{Boundary data.} We decompose the gravitino trace into normal and tangential
parts, and then split the tangential vector--spinor into its gamma-traceless and gamma-trace
pieces:
\begin{equation}
\psi_\mu\to(\psi_\perp,\psi_a),\qquad
\psi_a=\psi_a^{\gtf}+\frac13\gamma_a\chi,\qquad
\chi=\gamma^a\psi_a,\qquad
\gamma^a\psi_a^{\gtf}=0 .
\end{equation}
The bosonic conformal data and the trace-free tensor space used below are
\begin{equation}
B_{\mathrm{bos}}=\{\eqref{eq:york}\},\qquad
\Stf:=S^2_0(T^*\dd X).
\end{equation}

\paragraph{Symbol convention.} At fixed nonzero tangential momentum $k$, let $E_\psi$ be the
full vector--spinor trace space,
\begin{equation}
E_\psi=\C^{16}_{(\psi_\perp,\psi_1,\psi_2,\psi_3)} .
\end{equation}
After the local high-momentum reduction, a boundary condition is represented by a matrix
$b(k):E_\psi\to W(k)$, with $W(k)$ the boundary data space for the imposed rows; the boundary
condition is
\begin{equation}
b(k)\psi|=0 .
\end{equation}
In this paper
\[
V(k):=\ker b(k)\subset E_\psi
\]
is the \emph{allowed} or \emph{kept} boundary datum, physically, the half of the gravitino
boundary phase space the condition leaves free. Thus ``half-dimensional'' means
\begin{equation}
\dim_\C V(k)=8,\qquad \dim_\C E_\psi/V(k)=8 .
\end{equation}
This is the gravitino analogue of fixing half the components, as a Dirichlet/Neumann split
does for a bosonic field.
Equivalently, $V(k)$ is a point of the middle Grassmannian $\mathrm{Gr}(8,16)$ of the boundary
vector--spinor space. We do not require it to be isotropic (Lagrangian) for the boundary pairing
of the Dirac-type operator \cite{BaerBallmann:2012}, i.e.\ we do not assume the boundary
condition defines a self-adjoint realization; the no-go therefore covers a strictly larger class
than the self-adjoint/Lagrangian boundary conditions.
This convention is opposite to descriptions that call the image of $b(k)$ the killed data.
Local differential boundary conditions give polynomial or homogeneous dependence on $k$, while
spectral APS/Calderon-type conditions give nonlocal projectors depending on the direction
$\hat k$; the no-go below is pointwise in $k$, so it allows arbitrary half-dimensional $V(k)$, including
non-componentwise vector--spinor symbols and mixing between $\psi_\perp$ and the $\psi_a$.

Two notions of ``supersymmetric boundary condition'' will recur, and it is worth fixing them
at the outset. By a \emph{local supersymmetric conformal completion} we mean a half-dimensional
boundary condition $B_{\mathrm{ferm}}$ on the gravitino trace space $\psi_\mu$, supplemented by
compatible ghost conditions on $\eta$, all given by finite-order differential boundary operators, such that
\begin{equation}
(B_{\mathrm{bos}},B_{\mathrm{ferm}})
\quad\text{is LS-elliptic and BRST-compatible at principal-symbol level.}
\end{equation}
In other words, $B_{\mathrm{ferm}}$ is a gravitino condition admissible alongside the
conformal data and the gauge fixing. And we say a boundary condition $B$ is \emph{preserved by
the rigid half-supersymmetry} if
\begin{equation}
\delta_\varepsilon B|\in\operatorname{span}(B)+(\text{compensating gauge}),
\qquad
\Pp\varepsilon=\varepsilon,\quad \varepsilon\in\img\Pp .
\end{equation}
That is, the boundary-condition surface is $\delta_\varepsilon$-invariant modulo
diffeomorphism and local-supersymmetry transformations. (At this linearization about
$\psi^{(0)}=0$ local Lorentz acts trivially on
$(h,\psi)$, rotating only the frame, so the compensating quotients used below are the
boundary-diffeomorphism and local-supersymmetry images.)

We do not analyze the weaker problem of preserving a single fixed complex supercharge
$\varepsilon_0\in\img\Pp$. The matrix obstruction below uses the full half-space
$\img\Pp$.

\subsection{The forward direction forces the chiral projector}
\label{sec:forward}

Supersymmetry acts on the bosonic datum first, so we start there: which gravitino boundary
data is even consistent with holding the conformal class fixed under a supersymmetry
variation? The answer is rigid and follows from the Clifford algebra at the boundary: a
tangential gamma matrix anticommutes with the normal one, and therefore exchanges the two
boundary chiralities. Only the local chiral half of the gravitino survives this necessary
test; the self-adjoint spectral (APS) choice does not. The fully general half-dimensional symbol,
with boundary diffeomorphisms quotienting the forward variation, is treated in the
completion-independence argument below.

Only the conformal-class component $h_{ab}^{\tf}|=0$ is used in the forward analysis. The
mean-curvature and de~Donder rows in \eqref{eq:york} involve normal/tangential derivatives of
$h$, so their supersymmetry variations depend on derivatives of $\psi$ and on the detailed
mixed boundary complex. We do not claim that the chiral condition alone closes those rows.
For the theorem this necessary row is enough: forward closure of $h_{ab}^{\tf}|=0$ supplies
the constraint on the allowed gravitino data used below, and the reverse fermionic variation
then supplies the contradiction; no separate closure claim is made here for the mean-curvature
or de~Donder rows.

The claim is that exact closure of $\delta_\varepsilon h_{ab}^{\tf}|=0$, for a componentwise
half-dimensional spinor boundary datum, requires the kept tangential gravitino datum to lie in
$\img(\Pp)$: the local chiral projector achieves it exactly, while the spectral APS projector
does not. (Allowing closure only modulo boundary diffeomorphisms enlarges the forward kernel to
$F_{12}$, treated in the completion-independence argument below.) The mechanism is the boundary
Clifford relation: a tangential $\gamma_a$ anticommutes with $\gamma^\perp$, hence
\begin{equation}
\label{eq:flip}
\gamma_a\Pp=\Pm\gamma_a,\qquad \Pp\gamma_a=\gamma_a\Pm,\qquad
\Pp\gamma_a\Pp=\gamma_a\Pm\Pp=0,\qquad \Pp\gamma_a\Pm=\gamma_a\Pm.
\end{equation}
From \eqref{eq:susy},
\begin{equation}
\delta_\varepsilon h_{ab}
=\frac12\bar\varepsilon(\gamma_a\psi_b+\gamma_b\psi_a),
\qquad
\psi_b=\psi_b^{\gtf}+\frac13\gamma_b\chi .
\end{equation}
The $\chi$ contribution is pure trace:
\begin{equation}
\frac16\bar\varepsilon(\gamma_a\gamma_b+\gamma_b\gamma_a)\chi
=\frac13\delta_{ab}\bar\varepsilon\chi ,
\end{equation}
using $\{\gamma_a,\gamma_b\}=2\delta_{ab}$. Since
$\bar\varepsilon\gamma^a\psi_a^{\gtf}=0$, the trace-free part is
\begin{equation}
\delta_\varepsilon h_{ab}^{\tf}=\tfrac12\bar\varepsilon\big(\gamma_a\psi_b^{\gtf}+\gamma_b\psi_a^{\gtf}\big),
\end{equation}
automatically traceless and depending only on $\psi_a^{\gtf}$.

Let the gravitino condition keep the half-dimensional subspace $V\subset\C^4$ in the spinor index
(so the kept boundary datum is valued in $V$), and let $\varepsilon$ range over $\img\Pp$, so
$\bar\varepsilon=\bar\varepsilon\Pp$. Then exact vanishing of
$\delta_\varepsilon h_{ab}^{\tf}|$ for all kept data and all such $\varepsilon$ is equivalent to
\begin{equation}
\Pp\gamma_a v=0
\qquad
\text{for every }v\in V\text{ and every }a .
\end{equation}
By \eqref{eq:flip},
\begin{equation}
\Pp\gamma_a v=\gamma_a\Pm v .
\end{equation}
Since each $\gamma_a$ is invertible, this vanishes for all $a$ iff $\Pm v=0$, so
\begin{equation}
V\subseteq\img\Pp .
\end{equation}
The chiral choice $V=\img\Pp$ satisfies this identically since $\Pp\gamma_a\Pp=0$.

For the ordinary spinor APS symbol,
\begin{equation}
V=\img\Pi_+,\qquad
\Pi_+=\frac12(\mathbf1+A),\qquad
A=i\gamma^\perp\gamma^a\hat k_a .
\end{equation}
Because $A$ contains one $\gamma^\perp$ and one tangential factor,
\begin{equation}
A\gamma^\perp=-\gamma^\perp A .
\end{equation}
Thus $A$ interchanges $\img\Pp$ and $\img\Pm$; its $+1$ eigenspace therefore meets both and is
contained in neither, so $\img\Pi_+\not\subseteq\img\Pp$. Hence APS fails forward closure
(equivalently $\Pp\gamma_a\Pi_+\neq0$; the explicit norm $0.35$--$0.50$ is checked in
the ancillary scripts listed in the Reproducibility note).

\paragraph{The obstruction survives any half-dimensional datum.}
The forward test above used a componentwise datum; we now show the obstruction is
\emph{completion-independent}. Fix $k\neq0$. There is no half-dimensional $V\subset\C^{16}$ (kept
boundary gravitino datum) that closes both the forward variation of $h_{ab}^{\tf}|$ and the
reverse variation of the fermionic condition. Because the argument quantifies over arbitrary
$V$, it covers local and pseudodifferential (APS-type) completions alike. Concretely, a half-dimensional
boundary condition is the vanishing of a rank-$8$ boundary symbol,
\begin{equation}
b(k)\psi|=0,\qquad b(k):\C^{16}\to W(k),\qquad V=\ker b(k).
\end{equation}
Here $b(k)$ may be any local differential or pseudodifferential symbol, an APS/Calder\'on
projector, or one mixing $\psi_\perp$ with the $\psi_a$. Such a condition is preserved by the
rigid supersymmetry modulo a gauge image $G$ iff
\begin{equation}
b(k)\,\delta_\varepsilon\psi\in b(k)\,G
\qquad
\Longleftrightarrow
\qquad
\delta_\varepsilon\psi\in V+G
\end{equation}
for every allowed bosonic datum and preserved $\varepsilon$. This is a condition on the
subspace $V$ alone, identical for every choice of $b(k)$. The argument has
three pieces. Forward closure forces the tangential projection of $V$ into
a kernel $F_{12}$, defined in \eqref{eq:fwd}. Reverse closure, even after quotienting by the full
local-supersymmetry gauge image
\begin{equation}
G_\zeta:=\{ik_a\zeta:\zeta\in\C^4\},
\end{equation}
forces the
$K^{\tf}$ obstruction image $R_{K^{\tf}}$ into $V+G_\zeta$. Since the
$K^{\tf}$ obstruction is purely tangential, the two requirements imply
$R_{K^{\tf}}\subset F_{12}+G_\zeta$, and the explicit matrix below contradicts this.

Split $\C^{16}=\C^{12}_{\psi_a}\oplus\C^4_{\psi_\perp}$ and write
$V_{\mathrm{tang}}=\mathrm{proj}_{\mathrm{tang}}(V)$. The bosonic variation
$\delta_\varepsilon h_{ab}^{\tf}$ depends only on $\psi_a$ (as found above), so
forward closure modulo a compensating boundary diffeomorphism is
$\delta_\varepsilon h^{\tf}\in\img\mathrm{CK}(k)$, i.e.
\begin{equation}
\label{eq:fwd}
V_{\mathrm{tang}}\subseteq F_{12}:=\ker\!\big(\delta_\varepsilon h^{\tf}\bmod\img\mathrm{CK}(k)\big),
\end{equation}
where $\mathrm{CK}(k):\xi_a\mapsto(ik_{(a}\xi_{b)})^{\tf}$ is the conformal-Killing operator,
of rank $3$ in the $5$-dimensional $\Stf$ (an elementary check: $k_{(a}\xi_{b)}^{\tf}=0$ forces
$\xi\parallel k$ and then $\xi=0$).

The reverse variation $\delta_\varepsilon(\text{fermionic BC})$ contains the broken-connection
term whose gauge-invariant part spans $R_{K^{\tf}}\subset\C^{12}_{\psi_a}$ (identified with
$K_{ab}^{\tf}$ in Section~\ref{sec:reverse}); closure modulo the local-supersymmetry gauge
$G_\zeta=\{ik_a\zeta:\zeta\in\C^4\}$ requires $R_{K^{\tf}}\subseteq V+G_\zeta$. Here
$G_\zeta$ is deliberately the \emph{full} tangential local-supersymmetry gauge image, not the
smaller image obeying any particular ghost boundary condition such as $\Pm\eta|=0$ in
Section~\ref{sec:completion}. Quotienting by the full image makes the obstruction easier to
remove; failure after this enlarged quotient implies failure for every restricted
compensator allowed by a BRST-compatible highest-derivative boundary condition. As both $R_{K^{\tf}}$ and
$G_\zeta$ are purely tangential, $R_{K^{\tf}}\subseteq V_{\mathrm{tang}}+G_\zeta$; combined with
\eqref{eq:fwd},
\begin{equation}
\label{eq:joint}
R_{K^{\tf}}\subseteq F_{12}+G_\zeta.
\end{equation}
All three subspaces $F_{12}$, $G_\zeta$, and $R_{K^{\tf}}$ live in the tangential
$\C^{12}_{\psi_a}$, so \eqref{eq:joint} is a single inclusion of subspaces of $\C^{12}_{\psi_a}$.
``Closure for all preserved $\varepsilon$'' means each inclusion must hold simultaneously over a
basis of the two-dimensional preserved spinor space $\img\Pp$; the obstruction is therefore the
map stacked over those two basis spinors, which is why the matrix $M_{\tf}$ below has
$2\times2=4$ rows, the two transverse-traceless metric components for each of the two basis
$\varepsilon$. We now check \eqref{eq:joint} explicitly. By boundary $SO(3)$ covariance take
$\hat k=e_1$. For the four-dimensional spinor space $S$, write
$S=S_+\oplus S_-$ with $S_\pm=\img P_\pm^\perp$ and choose a boundary-chiral basis
\[
\gamma^\perp=\begin{pmatrix}1&0\\0&-1\end{pmatrix},\qquad
\gamma_a=\begin{pmatrix}0&\sigma_a\\ \sigma_a&0\end{pmatrix},
\]
with Pauli matrices $\sigma_a$. Let $\phi_a=\Pm\psi_a\in S_-$ and let
$u_1,u_2$ be the standard basis of $S_+$. The conformal-Killing image for $k=e_1$ is spanned
inside $\Stf$ by $(2,-1,-1)$, the $12$ component, and the $13$ component; hence the quotient
$\Stf/\img\mathrm{CK}$ is represented by the two transverse-traceless components
\[
q_1=\delta h_{22}-\delta h_{33},\qquad q_2=2\delta h_{23}.
\]
For a tangential gravitino datum $\phi=(\phi_1,\phi_2,\phi_3)$ and
$\varepsilon'\in S_+$,
\begin{equation}
\label{eq:ttres}
q_1(\varepsilon',\phi)=\varepsilon'^\dagger(\sigma_2\phi_2-\sigma_3\phi_3),\qquad
q_2(\varepsilon',\phi)=\varepsilon'^\dagger(\sigma_2\phi_3+\sigma_3\phi_2).
\end{equation}
As $\varepsilon'$ ranges over the two-dimensional $S_+$, these two TT components define the
quotient residual map $\C^{12}_{\psi_a}\to(\Stf/\img\mathrm{CK})\otimes S_+^*$; its rank is
$2$, so $\dim F_{12}=10$.
The local-supersymmetry gauge image $G_\zeta=\{i\delta_{a1}\zeta\}$ changes only $\phi_1$,
so it has already disappeared from \eqref{eq:ttres}. The reverse $K^{\tf}$ image is
\[
\phi_a=K_{ab}\sigma_b\varepsilon,\qquad K_{ab}=K_{ba},\qquad K_a{}^a=0,
\]
up to an irrelevant overall factor. Use the traceless basis
\[
D_1=\diag(1,-1,0),\quad D_2=\diag(1,1,-2),\quad E_{12},\quad E_{13},\quad E_{23}
\]
for $K^{\tf}$, where $E_{ij}$ denotes the symmetric off-diagonal matrix with unit
$ij$ and $ji$ entries. The columns below are ordered as $Bu_1,Bu_2$, with $B$ running through
this basis in the displayed order. With row order
$(q_1(u_1),q_1(u_2),q_2(u_1),q_2(u_2))$, \eqref{eq:ttres} gives
\begin{equation}
\label{eq:rk-matrix}
M_{\tf}=
\begin{pmatrix}
-1&0&3&0&-i&0&0&-1&0&2i\\
0&-1&0&3&0&i&1&0&2i&0\\
0&i&0&-3i&0&1&-i&0&2&0\\
i&0&-3i&0&-1&0&0&i&0&2
\end{pmatrix}.
\end{equation}
The last row is $-i$ times the first, the third row is $-i$ times the second, and the first
two rows are independent; hence $\rank M_{\tf}=2$. Equivalently, two independent combinations of
the $K_{ab}^{\tf}$ boundary data have no preimage in $F_{12}+G_\zeta$, so the projection of
$R_{K^{\tf}}$ to the quotient by $F_{12}+G_\zeta$ is nonzero (rank two), contradicting
\eqref{eq:joint}. Since any nonzero $k$ is rotated to $e_1$, the contradiction is pointwise
for every $k\neq0$. For reference, the dimensions entering this computation (at $k=e_1$) are
\[
\begin{array}{lcl}
\text{half-dimensional kept datum } V\subset\C^{16} & : & 8,\\
\text{forward kernel } F_{12}\subset\C^{12}_{\psi_a} & : & 10,\\
\text{local-supersymmetry gauge } G_\zeta\subset\C^{12}_{\psi_a} & : & 4,\\
\big(\Stf,\ \img\mathrm{CK},\ \Stf/\img\mathrm{CK}\big) & : & (5,\,3,\,2),\\
\rank M_{\tf}\ \big(R_{K^{\tf}}\bmod F_{12}+G_\zeta\big) & : & 2.
\end{array}
\]

Mixing in $\psi_\perp$ or in the gamma-trace $\chi$, and allowing $V(k)$ to depend
pseudodifferentially on the direction $\hat k$ (an APS/Calder\'on-type vector--spinor symbol),
does not remove the contradiction: the argument is pointwise in $k$, and \eqref{eq:joint} is
computed in the full tangential $\C^{12}_{\psi_a}$ space, including trace directions, while
$R_{K^{\tf}}$ itself has no normal component. Hence no half-dimensional $V(k)$ of any of these types
can contain $R_{K^{\tf}}$ modulo $F_{12}+G_\zeta$.

\subsection{The elliptic chiral completion at highest-derivative order}
\label{sec:completion}

Fixing the tangential gravitino data to be chiral still leaves its normal component and the
ghosts to be specified. The normal condition below contains $\dd_\perp\psi_\perp$. For a
bare first-order Dirac or Rarita--Schwinger operator this is not a standard local
self-adjoint boundary condition; for the general framework of mixed (chirality-split)
Dirichlet/Neumann boundary conditions for the Dirac operator see \cite{Luckock:1990xr}. Here it is used only in the local high-momentum sense
defined above: the mixed-order boundary symbol obtained after keeping the highest-derivative
terms. Equivalently, in PDE language, this is the Douglis--Nirenberg/Calderon symbol for the
squared Dirac-type normal component in the gauge-fixed complex. This section provides a
\emph{diagnostic} check: the natural chiral completion can satisfy the LS count and BRST
closure at highest-derivative order, so the later no-go is not caused by an elementary rank or
ghost mismatch. It is not an independent construction of a first-order self-adjoint
Rarita--Schwinger realization.

Concretely, the chiral completion used here at this order is
\begin{equation}
\label{eq:cand}
\Pm\psi_a\big|=0,\qquad \Pm\,\dd_\perp\psi_\perp\big|=0,\qquad \Pm\eta\big|=0,
\end{equation}
together with $B_{\mathrm{bos}}$ and Dirichlet ghosts $c_\mu|=0$. We claim it gives an
LS-elliptic boundary symbol (rank $24/24$) and is BRST-compatible at this order, and check
the two properties in turn. That a Dirichlet-type (chiral) gravitino boundary condition is
LS-elliptic is originally due to Massimo Porrati (private communication); the calculation
here makes this precise for the chiral/Robin completion in the gauge-fixed theory.

\paragraph{BRST closure.}
With $s$ denoting the BRST differential, $s\psi_\mu=\dd_\mu\eta$, and
with the ghost restricted to decaying modes,
\begin{equation}
\sd\eta=0,\qquad
\dd_\perp\eta=-\gamma^\perp\gamma^a\dd_a\eta ,
\end{equation}
the uniform chiral normal condition fails:
\begin{equation}
s(\Pm\psi_\perp|)=\Pm\dd_\perp\eta|=-\Pm\gamma^\perp\gamma^a\dd_a\eta|
=\Pm\gamma^a\dd_a\eta|=\gamma^a\dd_a(\Pp\eta)|\neq0,
\end{equation}
using $\Pm\gamma^\perp=-\Pm$, $\Pm\gamma^a=\gamma^a\Pp$, and $\Pp\eta|=\eta|$ (from
$\Pm\eta|=0$). The Robin condition closes, because on the same modes
\begin{equation}
\dd_\perp^2\eta
=\gamma^\perp\gamma^a\gamma^\perp\gamma^b\dd_a\dd_b\eta
=-\gamma^a\gamma^b\dd_a\dd_b\eta
=-\dd_\parallel^2\eta
=|k|^2\eta,\qquad
\dd_\parallel^2:=\delta^{ab}\dd_a\dd_b ,
\end{equation}
so
\begin{equation}
s(\Pm\dd_\perp\psi_\perp|)=\Pm\dd_\perp^2\eta|=|k|^2\,\Pm\eta|=0.
\end{equation}

\paragraph{Ellipticity.}
Again rotate $k$ to $e_1$ and take a decaying Laplace mode
$e^{iy^1-x^\perp}$. The conformal (bosonic) block below is exactly the elliptic Witten
calculation of Section~\ref{sec:bosonic-york}, now appearing as one of four block-diagonal
sectors; we use the same graviton amplitudes $H_{\mu\nu}$ (with $H_{\perp a}=q_a$,
$H_{\perp\perp}=\beta$, and common trace value $r$ as there). Write the ten boundary amplitudes
in the order
\[
(H_{\perp\perp},H_{\perp1},H_{\perp2},H_{\perp3},
H_{11},H_{22},H_{33},H_{12},H_{13},H_{23}).
\]
The five equations $H_{ab}^{\tf}=0$ give
$H_{11}=H_{22}=H_{33}=r$ and $H_{12}=H_{13}=H_{23}=0$. The two equations
$T_2=T_3=0$ give $H_{\perp2}=H_{\perp3}=0$. The remaining three equations reduce to
\begin{equation}
\label{eq:york-e1}
-2iH_{\perp1}-3r=0,\qquad
-\tfrac12 H_{\perp\perp}+iH_{\perp1}+\tfrac32r=0,\qquad
-\tfrac i2 H_{\perp\perp}-H_{\perp1}-\tfrac i2r=0,
\end{equation}
whose only solution is $r=H_{\perp1}=H_{\perp\perp}=0$. Thus the conformal block has
rank $10$ on the ten-dimensional decaying graviton space.

For the Dirac-type blocks the $+1$ eigenspace of
$A(e_1)=i\gamma^\perp\gamma_1$ is
\[
\{(u,-i\sigma_1u):u\in S_+\}\subset S_+\oplus S_-.
\]
Therefore $\Pm$ maps the decaying spinor space isomorphically to $S_-$, and
$\Pm\dd_\perp=-\Pm$ is also an isomorphism on decaying modes. The rank contributions are
\begin{equation}
\rank_{\psi}=6+2=8,\qquad
\rank_{\mathrm{ghost}}=4+2=6 .
\end{equation}
Here the three tangential conditions $\Pm\psi_a|=0$ contribute the $6$, the normal Robin
condition contributes the $2$, Dirichlet for the four Laplace ghosts $c_\mu$ contributes
$4$, and $\Pm\eta|=0$ contributes $2$ on the Dirac decaying ghost. In the local high-momentum
reduction, the boundary symbol is block diagonal at highest-derivative order (boson, diffeomorphism ghost, gravitino,
supersymmetry ghost), hence its rank is
\[
10+4+8+2=24
\]
on the $24$-dimensional decaying space. The Robin condition on $\psi_\perp$ is used here as
part of this mixed-order highest-derivative boundary problem (equivalently for the squared
Dirac-type normal component); by itself it is not being claimed to define a self-adjoint
first-order Rarita--Schwinger realization. This is enough for the diagnostic statement and
for the highest-derivative BRST check above.

\paragraph{Boundary pairing identity.}
The bare first-order Rarita--Schwinger boundary pairing contains,
with antisymmetrized Clifford products understood,
\[
\oint\bar\psi_a\gamma^{a\perp b}\psi_b .
\]
The Clifford identity
$(\gamma^{a\perp b}+\gamma^{ab})\Pp=0$ holds. For $a\neq b$,
\[
(\gamma^a\gamma^\perp\gamma^b+\gamma^a\gamma^b)\Pp
=\gamma^a(\gamma^\perp\gamma^b+\gamma^b)\Pp
=\gamma^a(-P_-^{\perp}\gamma^b+P_-^{\perp}\gamma^b)=0,
\]
using $\gamma^b\Pp=\Pm\gamma^b$ and $\gamma^\perp\Pm=-\Pm$; both sides vanish for
$a=b$. Hence on the chiral slice
\eqref{eq:cand} the Belyaev--van~Nieuwenhuizen-type boundary bilinear
$\tfrac12\oint\bar\psi_a\gamma^{ab}\psi_b$ has the right principal Clifford structure to
pair with the bulk boundary form (overall sign set by the bulk normalization). A complete
first-order variational and self-adjoint realization is not claimed here.

\subsection{The reverse direction and the broken-Lorentz obstruction}
\label{sec:reverse}

Now run supersymmetry the other way. Acting on the fermionic condition it must induce a
condition on the metric, and we ask whether that induced condition is compatible with the
conformal data. It is not, and the obstruction is the one piece of boundary geometry the
conformal prescription declines to fix, the traceless extrinsic curvature.

\paragraph{The obstruction is the normal--tangential spin connection.}
For the chiral completion \eqref{eq:cand}, $\delta_\varepsilon(\Pm\psi_a|)$ does not vanish on
$B_{\mathrm{bos}}$ data; its nonzero part is the boundary normal--tangential spin connection
$\omega_a{}^{b\perp}|=K_{ab}+\tfrac12 ik_a h_{b\perp}$, which $B_{\mathrm{bos}}$ leaves free.
To see this, split the gravitino variation into tangential--tangential and mixed pieces,
\begin{equation}
\delta_\varepsilon\psi_a
=\frac14\omega_a{}^{mn}\gamma_{mn}\varepsilon
=\frac14\omega_a{}^{bc}\gamma_{bc}\varepsilon
+\frac12\omega_a{}^{b\perp}\gamma_{b\perp}\varepsilon .
\end{equation}
Since $\gamma_{bc}$ (two tangential factors) commutes with $\gamma^\perp$,
$\gamma_{bc}\varepsilon\in\img\Pp$ for $\varepsilon\in\img\Pp$. The mixed factor has the
opposite chirality:
\begin{equation}
\gamma_{b\perp}=\gamma_b\gamma^\perp,\qquad
\gamma_{b\perp}\varepsilon
=\gamma_b\gamma^\perp\varepsilon
=\gamma_b\varepsilon\in\img\Pm .
\end{equation}
The chiral condition $\Pm\psi_a|=0$ keeps $\img\Pp$ and tests the complementary projection
$\Pm(\cdot)$, which annihilates the tangential--tangential part; hence
\begin{equation}
\label{eq:obs}
\delta_\varepsilon(\Pm\psi_a|)=\tfrac12\,\omega_a{}^{b\perp}|\,\gamma_b\varepsilon .
\end{equation}
The relevant linearized spin connection and extrinsic curvature are
\begin{equation}
\omega_\mu{}^{mn}
=\frac12(\dd^n h_\mu{}^m-\dd^m h_\mu{}^n),\qquad
K_{ab}
=\frac12(\dd_\perp h_{ab}-\dd_a h_{b\perp}-\dd_b h_{a\perp}) .
\end{equation}
Therefore
\begin{equation}
\omega_a{}^{b\perp}
=\frac12(\dd_\perp h_{ab}-\dd_b h_{a\perp})
=K_{ab}+\frac12\dd_a h_{b\perp},\qquad
\omega_a{}^{b\perp}|
=K_{ab}+\frac12ik_a h_{b\perp}| .
\end{equation}
Thus $B_{\mathrm{bos}}$ fixes $h_{ab}^{\tf}|$, $K|$ and $T_\mu|$, none of which constrains
$\dd_\perp h_{ab}^{\tf}|$ (i.e.\ $K_{ab}^{\tf}$) or $h_{a\perp}|$.

Moreover \eqref{eq:obs} vanishes for all preserved $\varepsilon$ if and only if
$\omega_a{}^{b\perp}|=0$. Indeed, in the chiral basis used above,
$\omega_a{}^{b\perp}\gamma_b\varepsilon$ is the Pauli matrix
$\big(\sum_b\omega_a{}^{b\perp}\sigma_b\big)\varepsilon$ acting from $S_+$ to $S_-$.
The matrices $\sigma_1,\sigma_2,\sigma_3$ are linearly independent, so this map is zero for
all $\varepsilon\in S_+$ exactly when every coefficient $\omega_a{}^{b\perp}$ is zero.
The conformal rows do not imply this. For $k=e_1$, the boundary jet
\begin{equation}
h|=0,\qquad
\dd_\perp h_{22}=1,\qquad
\dd_\perp h_{33}=-1 ,
\end{equation}
with all other normal derivatives zero, satisfies \eqref{eq:york} but has nonzero
$K^{\tf}$. Also the boundary value
\begin{equation}
h_{2\perp}|=1,\qquad
\omega_1{}^{2\perp}=\frac i2 ,
\end{equation}
with all other $h|$ and normal derivatives zero, satisfies \eqref{eq:york}. Thus the raw obstruction is precisely the
unfixed normal--tangential connection, with a gauge-invariant $K^{\tf}$ part and a shift
part.

\paragraph{The gauge-invariant core is $\Stf$.}
This raw obstruction splits into the $K_{ab}^{\tf}$ image ($5$-dimensional, gauge-invariant)
and the $h_{a\perp}$-dependent spinor image ($3$-dimensional, lying in the local-supersymmetry
gauge image), and the quotient by the local-supersymmetry
gauge is canonically the traceless extrinsic curvature $K_{ab}^{\tf}\in\Stf$. Indeed, by
\eqref{eq:obs} the obstruction is
\begin{equation}
\frac12\omega_a{}^{b\perp}|\gamma_b\varepsilon,\qquad
\omega_a{}^{b\perp}|=K_{ab}+\frac12ik_a h_{b\perp}| .
\end{equation}
The $h_{a\perp}$ contribution is of the local-supersymmetry gauge form,
\begin{equation}
\frac14ik_a(h_{b\perp}\gamma_b\varepsilon)
=ik_a\zeta,\qquad
\zeta=\frac14h_{b\perp}\gamma_b\varepsilon ,
\end{equation}
hence lies in $G_\zeta$ and is removable.

It remains to see that the traceless $K$ contribution injects into the quotient. By
$SO(3)$ covariance take $k=e_1$. If
\[
K_{ab}^{\tf}\gamma_b\varepsilon=ik_a\zeta(\varepsilon)
\]
for all preserved $\varepsilon$, then the $a=2,3$ components of the left side must vanish:
\[
K_{2b}^{\tf}\sigma_b\varepsilon=0,\qquad K_{3b}^{\tf}\sigma_b\varepsilon=0
\qquad(\varepsilon\in S_+).
\]
Since $\sigma_1,\sigma_2,\sigma_3$ are linearly independent as maps $S_+\to S_-$, this gives
$K_{2b}^{\tf}=K_{3b}^{\tf}=0$ for all $b$. Symmetry then gives
$K_{12}^{\tf}=K_{13}^{\tf}=K_{23}^{\tf}=0$ and $K_{22}^{\tf}=K_{33}^{\tf}=0$; tracelessness
forces $K_{11}^{\tf}=0$. Thus the kernel is zero. The trace $\delta^{ab}K_{ab}=K$ is fixed
by $B_{\mathrm{bos}}$ (it is $\tfrac12\mathcal K$, fixed by the row $\mathcal K|=0$), and the
surviving quotient is represented by $K_{ab}^{\tf}$:
\begin{equation}
\text{obstruction}/G_\zeta\cong\Stf .
\end{equation}

\subsection{Main theorem}

\paragraph{Scope of the theorem.}
The theorem is a pointwise statement about the local high-momentum, highest-derivative limit
of the complexified, linearized, gauge-fixed theory on flat half-space. Equivalently, it is a
statement about the principal boundary symbol defined in Section~\ref{sec:setup}. It quantifies over
all choices $V(k)$ of which half of the boundary vector--spinor data is left free at each
$k\neq0$. It does
not assert a global elliptic/Fredholm theorem, a nonlinear no-go theorem, a self-adjoint first-order
Rarita--Schwinger realization (the Robin condition on $\psi_\perp$ of Section~\ref{sec:completion} is a
mixed highest-derivative completion only), or a complete quantum BRST boundary complex.

\begin{theorem}
\label{thm:main}
For the linearized gauge-fixed Euclidean $\mathcal N=1$ supergravity model on $\R^4_+$ with
fixed conformal bosonic boundary data \eqref{eq:york}, there is no completion $V(k)$ that
leaves half of the gravitino boundary data free (local or pseudodifferential,
APS-type included, and with any accompanying ghost boundary condition compatible with BRST at
this highest-derivative level) that preserves the full rigid half-supersymmetry
$\Pp\varepsilon=\varepsilon$. The obstruction is the traceless extrinsic curvature
$K_{ab}^{\tf}\in\Stf$, the canonical momentum conjugate to the conformal class $[h_{ab}]$,
which \eqref{eq:york} leaves as free Neumann/response data while the half-supersymmetry ties
it to the Dirichlet-type chiral datum $\Pm\psi_a|$.
\end{theorem}

Here $V(k)$ means a half-dimensional subspace of the full boundary vector--spinor trace space
$\C^{16}_{(\psi_\perp,\psi_a)}$ at each nonzero tangential momentum. It may depend on
$\hat k$, may mix $\psi_\perp$ with the three $\psi_a$, and need not be componentwise in the
vector index; the proof is pointwise in $k$, so no smoothness of the projector is used.
Thus a genuine vector--spinor APS or Calderon symbol, if imposed pointwise in tangential
momentum, is included as one such choice of $V(k)$. The freedom to choose ghost boundary
conditions cannot weaken the obstruction: in the proof the compensating local-supersymmetry
variation is allowed to be as large as possible at highest-derivative order.

\begin{proof}
The contradiction is completion-independent, as established in Section~\ref{sec:forward}. Let $V$ be
the kept boundary datum of any half-dimensional completion. Forward closure of $h_{ab}^{\tf}|=0$,
modulo a compensating boundary diffeomorphism, requires
$\mathrm{proj}_{\mathrm{tang}}(V)\subseteq F_{12}$ as in Section~\ref{sec:forward}. Reverse closure of
the fermionic condition, even modulo the enlarged full local-supersymmetry gauge $G_\zeta$,
requires the gauge-invariant obstruction $R_{K^{\tf}}\subseteq V+G_\zeta$; by the
spin-connection analysis of Section~\ref{sec:reverse} this obstruction is the traceless extrinsic
curvature $K_{ab}^{\tf}\in\Stf$, which is purely tangential, so
$R_{K^{\tf}}\subseteq F_{12}+G_\zeta$. But $R_{K^{\tf}}\not\subseteq F_{12}+G_\zeta$
(Section~\ref{sec:forward}). Hence no half-dimensional completion, local or pseudodifferential, closes both
directions; equivalently the conformal class $[h_{ab}]$ and its conjugate momentum
$K_{ab}^{\tf}$ cannot both be fixed, which is the asserted Dirichlet--Neumann obstruction.

The completion of Section~\ref{sec:completion} enters not as a premise but as a diagnostic check:
the natural chiral/Robin completion \eqref{eq:cand} is LS-elliptic ($24/24$) and
BRST-compatible at highest-derivative order, so the failure of half-supersymmetry is not due
to ellipticity or BRST compatibility at that order. It is the $K_{ab}^{\tf}$
reverse-supersymmetry obstruction described above.
\end{proof}

\section{Comments on Other Dimensions}
\label{sec:otherd}

The theorem above is four-dimensional. This section records what survives of the same
local high-momentum mechanism in other dimensions. Its purpose is diagnostic rather than
classificatory: the four-dimensional theorem is the only no-go result proved here, while the
lower- and higher-dimensional cases help identify which part of the obstruction is tied to
the boundary tensor algebra and which part depends on the particular minimal supergravity
multiplet.

The bosonic part of the local high-momentum count is uniform. In the pure metric/gravitino channel, for boundary
dimension $d-1\ge2$ (with $d$ the bulk dimension), the reverse supersymmetry variation produces the traceless
extrinsic curvature $K_{ab}^{\tf}\in\Stf$ as its gauge-invariant datum; quotienting in addition
by the linearized boundary conformal-Killing map $\xi_a\mapsto(\dd_{(a}\xi_{b)})^{\tf}$ leaves the
residual
\begin{equation}
\label{eq:dcount}
\dim\big(\Stf/\img\,\mathrm{CK}\big)=\frac{d(d-3)}{2}\qquad(d\ge3),
\end{equation}
the number of massless spin-two polarizations in $d$ dimensions; the case $d=2$ is exceptional,
with $\Stf=0$. Thus $d=4$ is the lowest dimension in which the residual is nonzero. It is also,
among the cases considered here, the unique nondegenerate one in which the minimal supergravity
multiplet is exactly $(g_{\mu\nu},\psi_\mu)$ (the graviton and gravitino on-shell degrees of
freedom matching, $2=2$) so that no extra field enters the boundary analysis. Below four
dimensions the residual vanishes; above it the residual grows but the minimal multiplet must
carry additional fields, which re-enter the reverse variation. Only the bosonic count
\eqref{eq:dcount} transcribes uniformly; the spinorial and extra-field parts must be rechecked
dimension by dimension. The comments below are therefore at the level of this local
highest-derivative boundary analysis; they are remarks, not theorems, and none has the status of Theorem~\ref{thm:main}. They
indicate where the question is open and are supported by the dimension-check ancillary scripts
listed in the Reproducibility note.

\subsection{Two dimensions}
On a one-dimensional boundary $\Stf=0$: there is no traceless extrinsic curvature and the
obstruction is vacuous (pure two-dimensional Einstein gravity is in any case topological). The
physically substantive two-dimensional theory is dilaton gravity. In Jackiw--Teitelboim gravity
\cite{Jackiw:1984je,Teitelboim:1983ux}
the dilaton fixes the bulk geometry to constant curvature, the boundary condition fixes the
induced length and the boundary dilaton rather than conformal boundary data, and the dynamical
boundary datum is the mean curvature, whose on-shell action is the Schwarzian
\cite{Maldacena:2016upp}. Its supersymmetric version preserves a boundary supersymmetry through
the super-Schwarzian \cite{Fu:2016vas,Stanford:2019vob}, but this proceeds through the dilaton
multiplet and at the near-boundary (asymptotically AdS$_2$) level; it is a different
boundary-value problem, orthogonal to the finite-boundary $K_{ab}^{\tf}$ mechanism studied here.
Conformal boundary conditions do reach two dimensions by a separate route: dimensionally reduced,
they descend to a JT-type dilaton gravity in which the trace $K$ survives as the dynamical
(Schwarzian) datum \cite{Banihashemi:2025qqi}, while the traceless $K_{ab}^{\tf}$ on which our
obstruction lives has no one-dimensional analogue, consistent with its vacuousness here.

\subsection{Three dimensions}
On a two-dimensional boundary $\dim\Stf=2$ but \eqref{eq:dcount} gives $0$: the obstruction
vanishes although the tensor $K_{ab}^{\tf}$ does not. The conformal-Killing map
$\xi_a\mapsto(\dd_{(a}\xi_{b)})^{\tf}$ is surjective onto $\Stf=S^2_0(\R^2)$ at every $k\neq0$
(two-dimensional metrics are conformally flat), so the $K_{ab}^{\tf}$ residue that obstructs four
dimensions is pure boundary-diffeomorphism gauge here. In the local high-momentum test, the
three-dimensional analogue of the chiral/Robin completion of Section~\ref{sec:completion}
is LS-elliptic (the conformal boson rows have full rank, as does the chiral-plus-Robin fermion
block), its diffeomorphism- and local-supersymmetry-ghost conditions are consistent at
highest-derivative order, and the reverse variation \eqref{eq:obs} closes modulo gauge. A local
half-supersymmetric conformal boundary condition therefore passes, in three dimensions, the same
pointwise high-momentum tests that four dimensions fails; this is not a global existence theorem, but
reflects that the transverse-traceless residue of $K_{ab}^{\tf}$ is pure boundary-diffeomorphism
gauge on a two-dimensional boundary while it survives on a three-dimensional one.
Two caveats. Three-dimensional gravity has no local propagating modes \cite{Witten:1988hc}, so this elliptic statement
is formal, the genuine content being global (boundary super-Virasoro data, or the Chern--Simons
formulation of three-dimensional supergravity \cite{Achucarro:1986uwr}) and outside this local
high-momentum test. And this is the flat
half-space $\R^3_+$ with $\Lambda=0$, distinct from the asymptotically AdS$_3$ problem.

\subsection{More than four dimensions}
Above four dimensions the residual grows as \eqref{eq:dcount}, but the minimal supergravity
multiplet is no longer $(g,\psi)$, and the extra fields re-enter the reverse variation. The
pattern is a chirality selection rule. Here ``channel'' means the normal-chirality sector of
the boundary vector--spinor that a given bosonic field strength can reach in
$\delta_\varepsilon\psi_a$. In five dimensions the gravity multiplet is
$(g_{\mu\nu},\psi_\mu^i,A_\mu)$, with $A_\mu$ the graviphoton and $i$ the $SU(2)$ index
\cite{Zucker:1999ej}; its
field strength $F_{\mu\nu}=\dd_\mu A_\nu-\dd_\nu A_\mu$ enters $\delta_\varepsilon\psi_a$,
and its tangential (magnetic) part $F_{ab}$ reaches the same chiral vector--spinor channel as
$K_{ab}^{\tf}$, so the pure-$(g,\psi)$ count is incomplete. Under the standard
ensemble (conformal metric data together with a Dirichlet condition on $A_a$) the
magnetic datum that reaches this channel is fixed source data, while the free Neumann datum is the electric
$F_{a\perp}$, which has the opposite boundary chirality and does not reach the channel. This
chirality-channel test in the standard Dirichlet-$A_a$ ensemble finds no cancellation by the free electric
datum, so the four-dimensional obstruction appears to persist, with bosonic residual
$d(d-3)/2=5$; a theorem would require the full Maxwell--gravitino highest-derivative boundary problem, including the
variation of the $A_a$ datum. A possible evasion indicated by the test is a non-standard
boundary relation tying $K_{ab}^{\tf}$ to the self-dual magnetic curvature, a full-multiplet
supersymmetric conformal boundary condition whose ellipticity we do not analyze. The channel statements
are checked numerically by the five-dimensional ancillary script listed in the Reproducibility note.
Only this five-dimensional selection rule has been checked concretely; what follows is a
structural guide, not a theorem. The Clifford-algebra
heuristic suggested by the five-dimensional check is that the reach depends on form degree: for a
$p$-form potential the electric and magnetic boundary components are expected to occupy opposite
chirality channels, so the standard Dirichlet datum fixes the reaching channel for even
field-strength degree (the five-dimensional graviphoton) but not necessarily for odd degree, and
self-dual or Chern--Simons-constrained tensors (such as the self-dual two-form of
six-dimensional minimal supergravity \cite{Nishino:1984gk}) replace the Dirichlet/Neumann split by a chiral
polarization and must be treated case by case. We do not resolve these. The stable conclusion of
these high-momentum checks is that four dimensions is the unique nondegenerate dimension in which the conformal-boundary
obstruction is both nonzero and unconditional: extra multiplet fields are absent there, present
and potentially curative above, and the obstruction itself absent below.

\section{Discussion}
\label{sec:disc}

In this work we have shown that Witten's elliptic conformal boundary problem for
perturbative Euclidean gravity admits no separated linear half-supersymmetric completion in
minimal supergravity, and identified the obstruction as the traceless extrinsic curvature
$K_{ab}^{\tf}$. The mechanism is a Dirichlet--Neumann mismatch. The conformal boundary
condition fixes the conformal class $[h_{ab}]$ (Dirichlet on the conformal metric) and the
mean curvature $K$ (Dirichlet on the conformal factor's conjugate), leaving $K_{ab}^{\tf}$ as
the free Neumann response, the canonical momentum conjugate to $[h_{ab}]$. The preserved
half-supersymmetry maps the Dirichlet-type chiral gravitino datum $\Pm\psi_a|$ into a
variation of precisely $K_{ab}^{\tf}$. Since a field and its conjugate momentum cannot both be
fixed, supersymmetry and the conformal boundary choice impose incompatible (Dirichlet versus
Neumann) requirements on the same boundary tensor. This is the content of
Theorem~\ref{thm:main}, and the gauge-invariant-core analysis of Section~\ref{sec:reverse} makes the
conjugate datum explicit as an element of $\Stf$.

Witten's point was that a finite-boundary
Euclidean gravitational path integral requires an elliptic boundary-value problem before the
Gaussian fluctuation integral, one-loop determinant, or heat-kernel expansion can be trusted.
For pure gravity, the conformal boundary condition provides such a boundary problem. The theorem
shows that the same perturbative framework has no separated linear half-supersymmetric completion in
minimal supergravity while $K_{ab}^{\tf}$ is kept as the free conformal response. This is not
only the failure of a particular elliptic ansatz: the chiral/Robin completion of
Section~\ref{sec:completion} is elliptic and BRST-compatible at highest-derivative order, but
supersymmetry still maps it to the unfixed $K_{ab}^{\tf}$. Thus a supersymmetric
finite-boundary perturbative path integral, if it exists, must modify the strict conformal
ensemble (for example by imposing a supercovariant graph tying $K_{ab}^{\tf}$ to boundary
fermions) rather than merely appending an independent gravitino boundary condition to
Witten's bosonic data.

\paragraph{Supercurrent multiplet and the nonlinear resolution.} The pair
$(K_{ab}^{\tf},\,\Pm\psi_a)$ is the bottom of a boundary supermultiplet, the
extrinsic-curvature/supercurrent multiplet of Belyaev--van~Nieuwenhuizen and Belyaev--Pugh
\cite{Belyaev:2005rt,Belyaev:2007bg,Belyaev:2010as}. Beyond linear order the spin connection is promoted to
the supercovariant $\what{\omega}_a{}^{b\perp}=\omega_a{}^{b\perp}+\psi\psi$, and the boundary
action ($\oint K$ for gravity, $\oint\bar\psi\gamma\psi$ for the gravitino) ties $K_{ab}^{\tf}$
to a gravitino bilinear. The obstruction found here points to precisely the sort of
$O(\psi^2)$ supercovariant term that the linearized analysis cannot supply: at linear order
$K_{ab}^{\tf}$ is free, while the nonlinear boundary calculus is designed to close the
multiplet. This suggests that, within the local high-momentum boundary analysis, a
supersymmetric refinement of conformal data would have to evade the obstruction through
nonlinear or fermion-bilinear-twisted boundary data, not through a no-go for supergravity with
boundary in general. In modern $4d$ $\mathcal N=1$ terms the
pair $(K_{ab}^{\tf},\Pm\psi_a)$ is a boundary avatar of the bulk supercurrent multiplet
\cite{Komargodski:2010rb}, realized concretely as the boundary extrinsic-curvature multiplet
of Belyaev--Pugh \cite{Belyaev:2010as}; more generally, in a theory with boundary it is the
total-derivative term in the bulk supersymmetry variation that obstructs invariance, removable
only by a boundary action and only for a sub-superalgebra \cite{DiPietro:2015zia,Luckock:1991se}. Our
gravitational statement is the precise version of this: the bosonic partner that the surviving
supercharges would pair with the gravitino datum $\Pm\psi_a|$ is exactly the response
$K_{ab}^{\tf}$ that the conformal prescription must leave free.

\paragraph{Relation to standard half-BPS boundaries.} This obstruction should not be confused
with the familiar half-supersymmetric branes and walls of supergravity, where the boundary
equations include a condition on the normal--tangential
connection. A totally geodesic or orbifold wall, as in the Ho\v{r}ava--Witten construction
\cite{Horava:1996ma}, sets $K_{ab}|=0$ up to supercovariant corrections; a pure-tension
end-of-the-world brane \cite{Takayanagi:2011zk} obeys an Israel equation \cite{Israel:1966rt}
$K_{ab}\propto h_{ab}$,
so $K_{ab}^{\tf}|=0$; with boundary matter an Israel equation relates $K_{ab}^{\tf}|$ to the boundary stress
tensor, so it is not left as the independent conformal-boundary response. In the boundary supergravity of Belyaev--van~Nieuwenhuizen
\cite{Belyaev:2005rt,Belyaev:2007bg} the same statement is supersymmetrized: the controlled datum
is the supercovariant $\what{\omega}_a{}^{b\perp}$, an extrinsic-curvature multiplet whose bosonic
component is tied to boundary fermion bilinears. None of these is the conformal ensemble considered
here, in which the conformal class is fixed while $K_{ab}^{\tf}$ is left as a free Neumann
response. Equation~\eqref{eq:obs} is exactly the linearized point where the distinction appears:
half-supersymmetry closes when $\what{\omega}_a{}^{b\perp}|$ may be constrained, and the
conformal prescription considered here forbids imposing an independent constraint on $K_{ab}^{\tf}$. The
asymptotically anti-de~Sitter supersymmetric Dirichlet problem
\cite{Festuccia:2011ws,BenettiGenolini:2016tsn} likewise fixes a conformal class
and reports a stress-tensor response, but it is a renormalized source/response setting at
conformal infinity, distinct from the finite-boundary local high-momentum analysis treated here.

\paragraph{Relation to the spectral-versus-local debate.} D'Eath and Esposito \cite{DEath:1991opx} concluded that gravitino boundary conditions in one-loop
Euclidean supergravity should be spectral (APS); see \cite{Vassilevich:2003xt} for the heat-kernel framework and
\cite{BaerBallmann:2012} for Dirac boundary value problems. Theorem~\ref{thm:main} sharpens and partly
redirects this in the local high-momentum boundary test: the usual spectral/APS symbol is naturally self-adjoint for Dirac-type problems but \emph{breaks} the chirality
($\Pp\gamma_a\Pi_+\neq0$) that the conformal-class supersymmetry variation requires
(Section~\ref{sec:forward}); the natural local datum is the chiral projector. Any
self-adjoint or variational completion of that chiral datum would have to be supplied by
additional boundary terms or supercovariant boundary data, not by replacing it with the
ordinary APS symbol. The genuine obstruction to a
fully supersymmetric \emph{linearized} boundary condition compatible with the conformal data is not
self-adjointness but the Dirichlet--Neumann mismatch on $K_{ab}^{\tf}$. We do not analyze
global APS issues such as zero modes, eta-invariants \cite{Atiyah:1976jg,Atiyah:1976qjr}, or the full nonlocal Fredholm problem.

\paragraph{An active bosonic program.} The bosonic side of this problem has become an active
subject. Witten's ellipticity observation has been developed into a study of well-posed
one-parameter families of conformal boundary conditions and the subtle difference between
their Euclidean and Lorentzian stability \cite{Liu:2024ymn,Liu:2025xij}, of finite-boundary
``gravitational observatories'' and their boundary dynamics
\cite{Anninos:2023epi,Anninos:2024wpy,Anninos:2024xhc}, of the
associated boundary thermal effective actions \cite{Banihashemi:2025qqi}, and of their
extension to higher-curvature gravity \cite{Galante:2025emz}. In this setting
our result has a direct interpretation: the bosonic conformal boundary condition can
be well posed and thermodynamically rich, but any supersymmetrization of it must address the
$K_{ab}^{\tf}$ response that the same prescription leaves free. The conflict is structural,
not a defect of a particular fermionic completion.

\paragraph{Relation to the gravitational path integral and holography.} Like Witten's note
\cite{Witten:2018lgb}, this is a finite-boundary local high-momentum result, motivated by the
semiclassical gravitational path integral (with the boundary term and the microcanonical
sense in which fixed boundary geometry defines an ensemble being the standard ones
\cite{Gibbons:1976ue,Brown:1992bq,Hawking:1995fd}) rather than by holography. Still, the
conformal datum \eqref{eq:york} is analogous to what one fixes at the conformal boundary of an
asymptotically anti-de~Sitter space, with $K_{ab}^{\tf}$ playing the role of the holographic
stress tensor \cite{Brown:1992br,Balasubramanian:1999re,deHaro:2000vlm}; in that light the
obstruction is the flat-space version of the statement that a supersymmetric conformal boundary
is naturally a boundary supercurrent multiplet, not a metric condition supersymmetrized
component by component. The finite-boundary setting makes this constraint local and explicit.
In AdS/BCFT the boundary geometry and its
extrinsic-curvature response are dynamical data carried by an end-of-the-world brane
\cite{Takayanagi:2011zk}, and imposing a conformal rather than Dirichlet condition there involves
precisely the traceless-extrinsic-curvature / brane-bending modes \cite{Chu:2021mvq}; the same
finite-boundary geometric data underlies wedge holography \cite{Akal:2020wfl} and
finite-cutoff / $T\bar T$-type holography \cite{McGough:2016lol,Hartman:2018tkw}. In every such setting a
half-supersymmetric refinement would have to specify how the $K_{ab}^{\tf}$ response enters the
boundary supermultiplet, which is exactly the conflict identified here. Whether the
supersymmetric conformal-boundary path integral is well posed once the nonlinear boundary
multiplet is restored, and the extension to asymptotically anti-de~Sitter boundaries
($\Lambda<0$, with the corresponding Killing spinors \cite{Festuccia:2011ws} and supersymmetric
counterterms \cite{BenettiGenolini:2016tsn}), are left to future work.

\paragraph{Outlook: the super-York boundary-value problem.} The obstruction is not a failure of
ellipticity or BRST compatibility at highest-derivative order (the chiral/Robin completion of
Section~\ref{sec:completion} satisfies both in the mixed-order high-momentum test explained there) but the
Dirichlet--Neumann mismatch for $K_{ab}^{\tf}$. It identifies, rather than forecloses, the
super-York problem: the surviving supercharges pair the gravitino datum with $K_{ab}^{\tf}$, so the
theorem points to a boundary relation, or graph in field space, tying $K_{ab}^{\tf}$ to a boundary supercurrent multiplet (rather
than a metric condition supersymmetrized component by component) as the form a supersymmetric
conformal boundary condition would have to take to absorb the obstruction.
This sets a concrete program. Its most direct realization is nonlinear: promote
$\omega_a{}^{b\perp}$ to the supercovariant $\what{\omega}_a{}^{b\perp}$ above and build this
boundary relation
from the boundary action of Belyaev--van~Nieuwenhuizen and Belyaev--Pugh, then establish its
LS-ellipticity and full BRST invariance. The dimensional pattern of Section~\ref{sec:otherd}
shows which parts of the obstruction are special to four dimensions. A further direction is
the same problem at an asymptotically anti-de~Sitter boundary, where the same boundary
relation should organize supersymmetric holographic renormalization
\cite{Festuccia:2011ws,BenettiGenolini:2016tsn}. In each case, the pointwise high-momentum
theorem proved here is the linearized obstruction any construction must absorb, a constraint
any super-York boundary condition must satisfy.

\paragraph{Reproducibility.}
The three rank-sensitive claims are analytic in the text: the completion-independence
obstruction is the rank-two matrix \eqref{eq:rk-matrix}, the LS symbol rank $24/24$ is the block
calculation following \eqref{eq:york-e1}, and the jet-space obstruction is the Pauli
linear-independence argument after \eqref{eq:obs}. The ancillary Python/NumPy scripts reproduce
these finite-dimensional reductions numerically and representation-independently. Running
\texttt{python3 verify\_all.py} from the ancillary directory checks the four-dimensional
assertions over $55$ tangential momenta (three coordinate axes, two generic directions, $50$
random) and under a random unitary change of gamma representation, printing
\texttt{all assertions passed} on success.

The component scripts are as follows. The forward chirality/APS test of
Section~\ref{sec:forward} is in \path{forward_chirality_aps.py}. The shared four-dimensional
gamma-matrix, York-symbol, and chiral/Robin-completion routines used in
Section~\ref{sec:completion} are in \path{four_d_symbol_algebra.py}. The
completion-independence maps and checks around \eqref{eq:rk-matrix} are in
\path{completion_independence_maps.py}, \path{completion_independence_backbone.py}, and
\path{completion_independence_mod_gauge.py}. The spin-connection obstruction
\eqref{eq:obs} is checked in \path{spin_connection_obstruction.py}. The quotient identifying
the gauge-invariant core with $\Stf$ is checked in \path{reverse_obstruction_gauge_test.py}
and \path{ktf_quotient_rank.py}. The dimension comments of Section~\ref{sec:otherd} are
supported by \path{dimension_counts_d2_d3.py}, \path{dimension_three_completion.py}, and
\path{dimension_five_channel.py}. The scripts are reproducibility checks and guards against
sign or gamma-convention errors, not substitutes for the finite-dimensional proofs above.

\addsec{Acknowledgements}
The author is grateful to Massimo Porrati for bringing Witten's paper to the author's
attention in 2022 and for showing that the Dirichlet boundary condition for the gravitino is
elliptic; this inspired the present work, and the author thanks him for valuable discussions.
The author also thanks Pinak Banerjee, Edgar Shaghoulian, and Edward Witten for
helpful discussions. The numerical verification scripts listed in the Reproducibility note
and part of the manuscript preparation were carried out with the assistance of AI tools,
Anthropic's Claude Code and OpenAI's Codex; the analytic proofs and all conclusions are the
author's own, and the author has checked the scripts and is responsible for any errors.
The author thanks Anthropic for providing access to the Claude Max plan. This work is
partially supported by the NSF grant PHY-2310588.

\printbibliography[heading=bibliography]

\end{document}